\documentclass[submission,copyright,creativecommons]{eptcs}
\providecommand{\event}{WORDS 2011} 

\usepackage{graphicx}
\usepackage{amsthm}
\usepackage{amssymb}
\usepackage{breakurl}

\newcommand{\Suff}{\textit{Suff}}
\newcommand{\Pref}{\textit{Pref}}
\newcommand{\Fact}{\textit{Fact}}

\renewcommand{\alph}{\textit{alph}}

\renewcommand{\epsilon}{\varepsilon}

\begin{document}

\title{A Classification of Trapezoidal Words}

\author{Gabriele Fici 
\institute{Laboratoire I3S, CNRS \& Universit\'e de Nice-Sophia Antipolis\\ 2000, Route des Lucioles - 06903 Sophia Antipolis cedex, France}  
\email{fici@i3s.unice.fr}
}

\newtheorem{theorem}{Theorem}
\newtheorem{proposition}[theorem]{Proposition}
\newtheorem{lemma}[theorem]{Lemma}
\newtheorem{corollary}[theorem]{Corollary}
\newtheorem{definition}{Definition}
\newtheorem{example}{Example}
\newtheorem{problem}{Problem}
\newtheorem{remark}{Remark}

\def\titlerunning{A Classification of Trapezoidal Words}
\def\authorrunning{G. Fici}

\maketitle

\begin{abstract}
Trapezoidal words are finite words having at most $n+1$ distinct factors of length $n$, for every $n\ge 0$. They encompass finite Sturmian words. We distinguish trapezoidal words into two disjoint subsets: open and closed trapezoidal words. A trapezoidal word is \emph{closed} if its longest repeated prefix has exactly two occurrences in the word, the second one being a suffix of the word. Otherwise it is \emph{open}. We show that open trapezoidal words are all primitive and that closed trapezoidal words are all Sturmian. We then show that trapezoidal palindromes are closed (and therefore Sturmian). This allows us to characterize the special factors of Sturmian palindromes. We end with several open problems.
\end{abstract}


\section{Introduction}\label{sec:intro}

In combinatorics on words, the most famous class of infinite words is certainly that of Sturmian words. Sturmian words code digital straight line in the discrete plane having irrational slope. They are characterized by the fact that they have exactly $n+1$ factors of length $n$, for every $n\ge 0$.

It is well known (\cite{LothaireAlg}, Proposition 2.1.17) that a finite word $w$ is a factor of some Sturmian word if and only if $w$ is a binary \emph{balanced} word, that is, there exists a letter $a$ such that for every pair of factors of $w$  of the same length, $u$ and $v$, one has that $u$ and $v$ contain the same number of $a$'s up to one, i.e.,
\begin{equation}\label{eq:bal}
 ||u|_{a}-|v|_{a}|\le 1.
\end{equation}

Finite Sturmian words (finite factors of Sturmian words) have the property that they have \emph{at most} $n+1$ factors of length $n$, for every $n\ge 0$. However, this property does not characterize them, as shown by the word $w=aaabab$, which is not Sturmian since the factors $aaa$ and $bab$ do not verify (\ref{eq:bal}).

The set of finite words defined by the property that they have at most $n+1$ factors of length $n$, for every $n\ge 0$, is called the set of \emph{trapezoidal words}. 

The name comes from the fact that the graph of the complexity function of these words\footnote{The complexity function of a word $w$ is the function that counts the number of distinct factors of each length in $w$.}  defines a regular trapezium. Trapezoidal words have been defined by de Luca, who observed that Sturmian words are trapezoidal \cite{Del99}. The non-Sturmian trapezoidal words have been then characterized by D'Alessandro \cite{Dal02}.

In this paper, we distinguish trapezoidal words into two distinct classes, accordingly with the definition below.

\begin{definition}\label{def:closed}
 Let $w$ be a finite word over an alphabet $\Sigma$. We say that $w$ is \emph{closed} if the longest repeated prefix of $w$ has exactly two occurrences in $w$, the second one being a suffix of $w$.
 
 A word which is not closed is called \emph{open}.
\end{definition}

For example, the word $aabbaa$ is closed, whereas the word $aabbaaa$ is open. 

\begin{remark}\label{rem:pl}
The notion of closed word is equivalent to that of \emph{periodic-like} word \cite{CaDel01a}. A word $w$ is periodic-like if its longest repeated prefix does not appear in $w$ followed by different letters.

The notion of closed word is also related to the concept of \emph{complete return} to a factor $u$ in a word $w$, as considered in \cite{GlJuWiZa09}. A complete return to $u$ in $w$ is any factor of $w$ having exactly two occurrences of $u$, one as a prefix and one as a suffix. Therefore, $w$ is a closed word if and only if $w$ is a complete return to its longest repeated prefix.

\end{remark}

In this paper, we distinguish trapezoidal words in open and closed. This allows us to establish some further properties of trapezoidal words. More precisely, we have that open trapezoidal words are all primitive (Lemma \ref{lem:oprim}), while closed trapezoidal words are all Sturmian (Proposition \ref{prop:cloStur}). We characterize open trapezoidal words by means of their special factors (Proposition \ref{prop:openspe}) and show that the longest special factor of a closed trapezoidal word is a central word (Lemma \ref{lem:clo1}). We then show that trapezoidal palindromes are closed (Theorem \ref{theor:palclo}). This allows us to characterize the special factors of Sturmian palindromes (Corollary \ref{cor:sturpal}).

\section{Trapezoidal Words}\label{sec:nota}

An \textit{alphabet}, denoted by $\Sigma$, is a finite set of symbols. A \textit{word} over $\Sigma$ is a finite sequence of symbols from $\Sigma$. We denote by $\alph(w)$ the subset of the alphabet $\Sigma$ constituted by the letters appearing in $w$. 

The \textit{length} of a word $w$ is denoted by $|w|$ and is the number of its symbols. We denote by $w_{i}$ the $i$-th letter of $w$. For a letter $a\in \Sigma$, we denote by $|w|_{a}$ the number of $a$'s appearing in $w$.
 
 The set of all words over $\Sigma$ is denoted by $\Sigma^*$. The set of all words over $\Sigma$ having length $n$ is denoted by $\Sigma^n$. The empty word has length zero and is denoted by $\varepsilon$. 

Let $w=a_1a_2\cdots a_n$, $n>0$, be a non-empty word over the alphabet $\Sigma$. The word $\tilde{w}=a_{n}a_{n-1}\cdots a_{1}$ is called the \emph{reversal} of $w$. A \emph{palindrome} is a word $w$ such that $\tilde{w}=w$.

A \textit{prefix} of $w$ is any word $v$ such that $v=\varepsilon$ or $v$ is of the form $v=a_1a_2\cdots a_i$, with $1\leq i\leq n$. 
A \textit{suffix} of $w$ is any word $v$ such that $v=\varepsilon$ or $v$ is of the form $v=a_ia_{i+1}\cdots a_n$, with $1\leq i\leq n$.  
A \textit{factor} of $w$ is a prefix of a suffix of $w$ (or, equivalently, a suffix of a
prefix of $w$). Therefore, a factor of $w$ is any word $v$ such that $v=\varepsilon$ or $v$ is of the form $v=a_ia_{i+1}\cdots a_j$, with $1\leq i\leq j\leq n$. A factor of a word $w$ is \emph{internal} if it is not a prefix nor a suffix of $w$.

We denote by $\Pref(w)$, $\Suff(w)$ and $\Fact(w)$, respectively, the set of prefixes, suffixes and factors of the word $w$.

The \emph{factor complexity} of a word $w$ is the function defined by $f_{w}(n)=|\Fact(w)\cap \Sigma^n|$, for every $n\geq 0$. Notice that $f_{w}(1)$ is the number of distinct letters occurring in $w$. A \emph{binary word} is a word $w$ such that $f_{w}(1)=|\alph(w)|=2$.

A factor $u$ of $w$ is \emph{left special} if there exist $a,b\in \Sigma$, $a\neq b$, such that $au,bu\in \Fact(w)$. A factor $u$ of $w$ is \emph{right special} if there exist $a,b\in \Sigma$, $a\neq b$, such that $ua,ub\in \Fact(w)$. A factor $u$ of $w$ is \emph{bispecial} if it is both left and right special. 

For example, let $w=aabbb$. The left special factors of $w$ are $\epsilon$, $b$ and $bb$. The right special factors of $w$ are $\epsilon$ and $a$. Therefore, the only bispecial factor of $w$ is $\epsilon$.

A \textit{period} for the word $w$ is a positive integer $p$, with $0<p\leq |w|$, such that
$w_{i}=w_{i+p}$ for every $i=1,\ldots ,|w|-p$. Since $|w|$ is always a period for $w$, we have that every non-empty word has at least one period. We can unambiguously define \textit{the} period of the word $w$ as the smallest of its periods. For example the period of $w=aabaaba$ is 3.

The \emph{fractional root} $z_{w}$ of a word $w$ is its prefix whose length is equal to the 
period of $w$.  So for example the fractional root of $w=aabaaba$ is $z_{w}=aab$.

A word $w$ is \emph{a power} if there exists a non-empty word $u$ and an integer $n>1$ such that $w=u^{n}$. A word which is not a power is called \emph{primitive}.

The following parameters have been introduced by de Luca \cite{Del99}:

\begin{definition} Let $w$ be a word over $\Sigma$. We denote by $H_w$ the minimal length of a prefix of $w$ which occurs only once in $w$. We denote by $K_w$ the minimal length of a suffix of $w$ which occurs only once in $w$.
\end{definition}

\begin{definition}
Let $w$ be a word over $\Sigma$. We denote by $L_w$ the minimal length for which there are no left special factors of that length in $w$. Analogously, we denote by $R_w$ the minimal length for which there are no right special factors of that length in $w$.
\end{definition}

\begin{example}\label{ex:1}
Let $w=aaababa$. The longest left special factor of $w$ is $aba$, and it is also the longest repeated suffix of $w$; the longest right special factor of $w$ is $aa$, and it is also the longest repeated prefix of $w$. Thus, we have $L_{w}=4$, $K_{w}=4$, $R_{w}=3$ and $H_{w}=3$.
\end{example}

Notice that for every word $w$ such that $|\alph(w)|> 1$, the values $H_{w},K_{w},L_{w}$ and $R_{w}$ are positive integers. Moreover, one has $f_{w}(R_{w})=f_{w}(L_{w})$ and $\max\{R_{w},K_{w}\}=\max\{L_{w},H_{w}\}$ (\cite{Del99}, Corollary 4.1).

The following proposition is from de Luca (\cite{Del99}, Proposition 4.2).

\begin{proposition}\label{prop:del}
Let w be a word of length $|w|$ such that $|\alph(w)|> 1$ and set
$m_{w} = \min\{R_{w}, K_{w}\}$ and $M_{w} = \max\{R_{w}, K_{w}\}$. The factor complexity $f_{w}$ is strictly increasing
in the interval $[0,m_{w}]$, is nondecreasing in the interval $[m_{w},M_{w}]$ and strictly decreasing
in the interval $[M_{w}, |w|]$. Moreover, for $i$ in the interval $[M_{w},|w|]$, one has $f_{w}(i + 1) = f_{w}(i) - 1$. If $R_{w}<K_{w}$, then $f_{w}$ is constant in the interval $[m_{w},M_{w}]$.
\end{proposition}

Proposition \ref{prop:del} allows one to give the following definition.

\begin{definition}
A non-empty word $w$ is \emph{trapezoidal} if 

\begin{itemize}
 \item $f_{w}(i)=i+1$ \ \ for \ \ $0\le i \le m_{w}$,
 \item $f_{w}(i+1)=f_{w}(i)$ \ \ for \ \  $m_{w} \le i \le M_{w}-1$,
 \item $f_{w}(i+1)=f_{w}(i)-1$ \ \ for \ \  $M_{w} \le i \le|w|$.
\end{itemize}
\end{definition}

Trapezoidal words have been considered for the first time by de Luca \cite{Del99}.  
The name \emph{trapezoidal} has been given by D'Alessandro \cite{Dal02}. The choice of the name is motivated by the fact that for these words the graph of the complexity function defines a regular trapezium (possibly degenerated in a triangle). 

\begin{figure}
\begin{center}
\includegraphics[height=60mm]{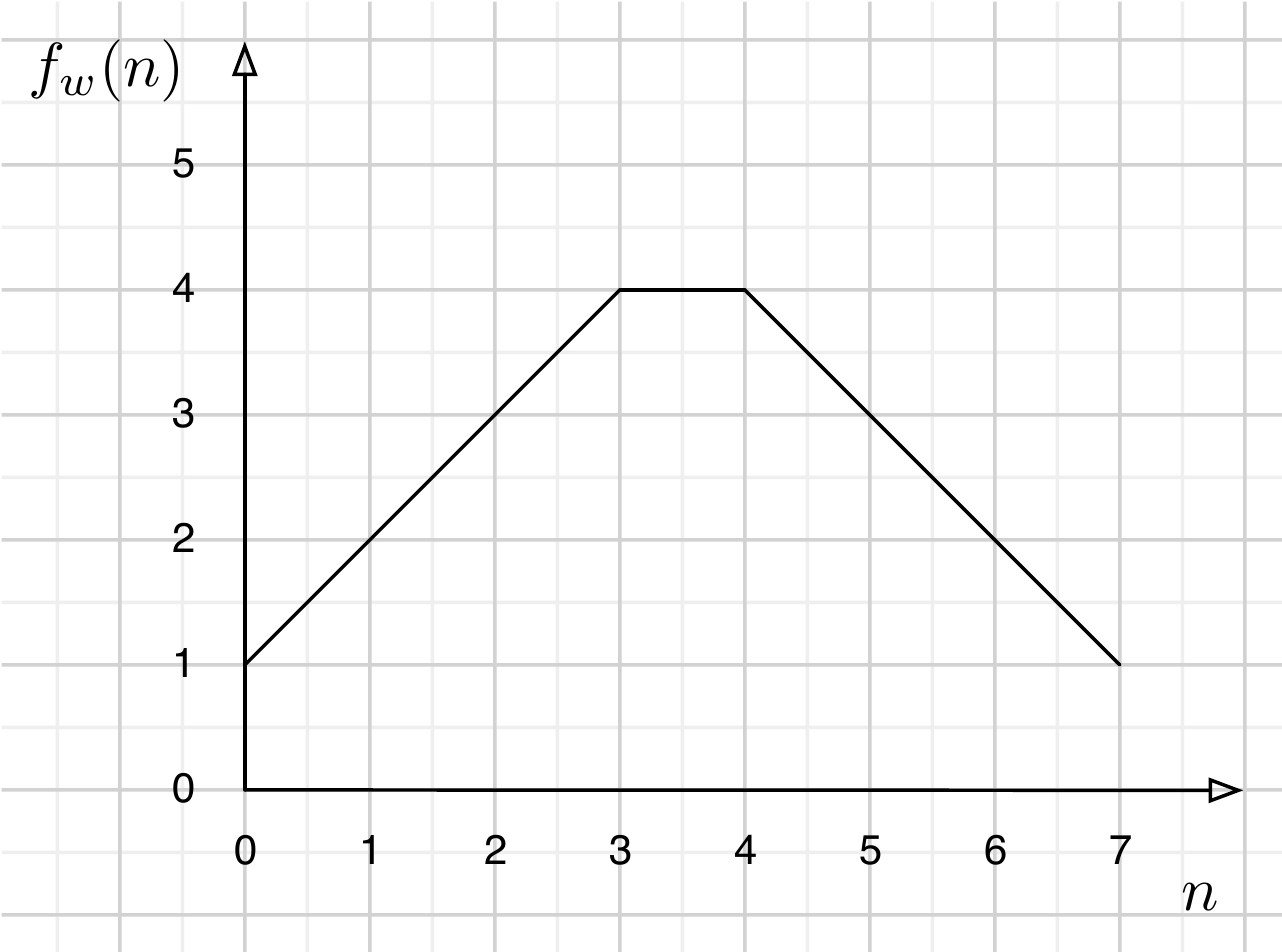}
\caption{The graph of the complexity function of the trapezoidal word $w=aaababa$. One has $m_{w}=\min \{R_{w},K_{w}\}=3$ and $M_{w}=\max \{R_{w},K_{w}\}=4$.}
\label{fig:graph}
\end{center}
\end{figure}

Notice that by definition a trapezoidal word is a binary word. 

\begin{example}
 The word $w=aaababa$ considered in Example \ref{ex:1} is trapezoidal. See Fig.\ \ref{fig:graph}.
\end{example}

In the following proposition we gather some characterizations of trapezoidal words. 

\begin{proposition}\label{prop:trap}
Let $w$ be a binary word. The following conditions are equivalent:
 
\begin{enumerate}
\item $w$ is trapezoidal;
 \item $|w|=L_{w}+H_{w}$;
 \item $|w|=R_{w}+K_{w}$;
 \item $w$ has at most one left special factor of length $n$ for every $n\geq 0$;
 \item $w$ has at most one right special factor of length $n$ for every $n\geq 0$;
 \item $w$ has at most $n+1$ distinct factors of length $n$ for every $n\geq 0$;
 \item $|f_{w}(n+1)-f_{w}(n)|\leq 1$ for every $n\geq 0$.
\end{enumerate}
\end{proposition}

\begin{proof}
The equivalence of $(1)$, $(2)$, $(3)$, $(4)$ and $(5)$ is in \cite{Del99} and \cite{Dal02}. 

The equivalence of $(5)$ and $(6)$ follows from elementary considerations on the factorial complexity of binary words. Indeed, it is easy to see the number of distinct factors of length $n+1$ of a binary word $w$ is at most equal to the number of distinct factors of length $n$ plus the number of right special factors of length $n$. 

The equivalence of $(7)$ and $(1)$ comes directly from the definitions and from Proposition \ref{prop:del}. 

The proof is therefore complete.
\end{proof}

Recall that a finite word is \emph{Sturmian} if  and only if it is balanced, i.e., verifies (\ref{eq:bal}). The following proposition is from de Luca (\cite{Del99}, Proposition 7.1).

\begin{proposition}\label{prop:sturmtrap}
Let $w$ be a binary word. If $w$ is Sturmian, then $w$ is trapezoidal. 
\end{proposition}

The inclusion in Proposition \ref{prop:sturmtrap} is strict, since there exist trapezoidal words that are not Sturmian, e.g. the word $w=aaababa$ considered in Example \ref{ex:1}. 

Recall that a binary word $w$ is \emph{rich} (or \emph{full}) \cite{GlJuWiZa09} if it contains $|w|+1$ distinct palindromic factors, that is the maximum number of distinct palindromic factors a word can contain. 
 
The following proposition is from de Luca, Glen and Zamboni  (\cite{DelGlZa08}, Proposition 2).

\begin{proposition}\label{prop:traprich}
Let $w$ be a binary word. If $w$ is  trapezoidal, then $w$ is rich.
\end{proposition}

Again, the inclusion in Proposition \ref{prop:traprich} is strict, since there exist rich words that are not trapezoidal, e.g. the word $w=aabbaa$.

D'Alessandro \cite{Dal02} characterized the non-Sturmian trapezoidal words. We report below his characterization.

First, recall that a word $w$ is unbalanced (i.e., $w$ is not balanced) if and only if it contains a pair of \emph{pathological factors} $(f,g)$, that is, $f$ and $g$ are factors of $w$ of the same length but they do not verify (\ref{eq:bal}). Moreover, if $f$ and $g$ are chosen of minimal length, there exists a palindrome $u$ such that $f=aua$ and $g=bub$, for two different letters $a$ and $b$ (see \cite{LothaireAlg}, Proposition 2.1.3). 

We can also state that $f$ and $g$ are Sturmian words, since otherwise they would contain a  pair of pathological factors shorter than $|f|=|g|$ and hence $w$ would contain such a pair of pathological factors, against the minimality of $f$ and $g$. So the word $u$ is a palindrome such that $aua$ and $bub$ are Sturmian words, i.e., $u$ is a \emph{central} word \cite{DelMi94}.

The following lemma is attributed to Aldo de Luca in \cite{Dal02}.

\begin{lemma}\label{lem:separation}
 Let $w$ be a non-Sturmian word and $(f,g)$ the pair of pathological factors of $w$ of minimal length. Then $f$ and $g$ do not overlap in $w$.
\end{lemma}

The following is the characterization of trapezoidal non-Sturmian words given by D'Alessandro \cite{Dal02}.

\begin{theorem}\label{theor:dal}
 Let $w$ be a binary non-Sturmian word. Then w is trapezoidal if and only if
$$w = pq, \mbox{\hspace{4mm} with $p \in \Suff(\{\tilde{z}_{f}^{*}\})$,\hspace{4mm} $q\in \Pref(\{z_{g}^{*}\})$}  $$
where $\tilde{z}_{f}$ is the mirror image of the fractional root $z_{f}$ of $f$, $z_{g}$ is the fractional root of $g$, with $(f,g)$ being the pair of pathological factors of $w$ of minimal length. 

In particular, $K_{w}=|q|$ and the longest right special factor of $w$ is the prefix of $w$ of length $R_{w}-1$.
\end{theorem}


\begin{example}
 Let $w=aaababa$ be the non-Sturmian trapezoidal word considered in Example \ref{ex:1}. We have  $f=aaa$ and $g=bab$, so that $\tilde{z}_{f}=a$ and $z_{g}=ba$. The word $w$ factorizes as $w=pq$, with $p=aaa$ and $q=baba$. 
\end{example}

Hence, trapezoidal words are either Sturmian or of the form described in Theorem \ref{theor:dal}. The following result of de Luca, Glen and Zamboni states that trapezoidal palindromes are all Sturmian. 

\begin{theorem}[\cite{DelGlZa08}]\label{theor:trappal}
 The following conditions are equivalent:
 
\begin{enumerate}
 \item $w$ is a trapezoidal palindrome;
 \item $w$ is a Sturmian palindrome.
\end{enumerate}
\end{theorem}

Let us give a proof of this latter result based on Theorem \ref{theor:dal}. We first show that the words $p$ and $q$ in the factorization of Theorem \ref{theor:dal} are Sturmian words.

\begin{lemma}\label{lem:pk}
Let $w$ be a trapezoidal non-Sturmian word and let $w=pq$, with $p \in \Suff(\{\tilde{z}_{f}^{*}\})$ and $q\in \Pref(\{z_{g}^{*}\})$, be the factorization of Theorem \ref{theor:dal}. Then $p$ and $q$ are Sturmian words. 
\end{lemma}

\begin{proof}
Recall that any central word $u$ that is not a power of a single letter can be uniquely written as $u=w_{1}xyw_{2}=w_{2}yxw_{1}$, for two central words $w_{1},w_{2}$ and different letters $x,y$ (see \cite{CarDel05}, Proposition 1).

Let $u$ be the central word such that $f=aua$ and $g=bub$. If $u$ is not a power of a single letter, the fractional roots of $f$ and $g$ are $xw_{2}y$ and $yw_{1}x$ (\cite{Dal02}, Lemma 2). This implies that $\tilde{z}_{f}$ and $z_{g}$ both are conjugate to standard Sturmian words\footnote{Standard Sturmian words are words of length one or of the form $uxy$, with $u$ central word and $x,y$ different letters \cite{DelMi94}.}. If  $u=x^{n}$, $x\in \Sigma$, $n>0$, then the fractional roots of $f$ and $g$ are $x$ and $yx^{n}$, so even in this case $\tilde{z}_{f}$ and $z_{g}$ both are conjugate to standard Sturmian words.

By Theorem 1 in \cite{DelDel06}, any word whose fractional root is conjugate to a standard Sturmian word is a Sturmian word. This implies that every word belonging to $\tilde{z}_{f}^{*}$ or to $z_{g}^{*}$ is Sturmian. Since a factor of a Sturmian word is a Sturmian word, the claim follows.
\end{proof}

Now, let $w$ be a trapezoidal palindrome. If $w$ is non-Sturmian we can write, by Theorem \ref{theor:dal}, $w=pq=v_{1}fgv_{2}=\tilde{v}_{2}gf\tilde{v}_{1}$, with $v_{1},v_{2}\in \{a,b\}^{*}$ such that $p=v_{1}f$ and $q=gv_{2}$. If $|v_{1}|=|v_{2}|$, then $f=g$, a contradiction. If $|v_{1}|\neq |v_{2}|$ then either $f$ and $g$ overlap (a contradiction with Lemma \ref{lem:separation}) or $p$ (or $q$) contains $f$ and $g$ as factors, a contradiction since, by Lemma \ref{lem:pk}, $p$ and $q$ are Sturmian. So $w$ cannot be a non-Sturmian word. 

Hence, we proved that trapezoidal palindromes are Sturmian. Since by Proposition \ref{prop:sturmtrap}, Sturmian words are trapezoidal, the claim of Theorem \ref{theor:trappal} follows.

\section{Open and Closed Trapezoidal Words}

In this section we derive some properties of open and closed trapezoidal words.

\begin{proposition}\label{prop:rev}
Let $w$ be a trapezoidal word. Then $w$ is open (resp.\ closed) if and only if $\tilde{w}$ is open (resp.\ closed).
\end{proposition}

\begin{proof}
If $w$ is a closed trapezoidal word, then its longest repeated prefix is also its longest repeated suffix and has exactly two occurrences in $w$. This implies that $\tilde{w}$ has the same property. So $\tilde{w}$ is a closed trapezoidal word. Hence the set of closed trapezoidal words is closed by reversal.  

Since the whole set of trapezoidal words is closed by reversal (\cite{Dal02}, Corollary 7) and since open trapezoidal words form the complement of closed trapezoidal words in the set of trapezoidal words, the set of open trapezoidal words is closed by reversal too.
\end{proof}

The following proposition gives a characterization of open trapezoidal words.

\begin{proposition}\label{prop:openspe}
 Let $w$ be a trapezoidal word. Then the following conditions are equivalent:
\begin{enumerate}
\item $w$ is open;
 \item the longest repeated prefix of $w$ is also the longest right special factor of $w$;
 \item the longest repeated suffix of $w$ is also the longest left special factor of $w$.
\end{enumerate}
\end{proposition}

\begin{proof}
$(1)\Rightarrow (2)$. Let $h$ be the longest repeated prefix of $w$ and $x$ the letter such that $hx$ is a prefix of $w$. Since $w$ is not closed, $h$ has a second non-suffix occurrence in $w$ followed by letter $y$. Since $h$ is the longest repeated prefix of $w$, we have $y\neq x$. Therefore, $h$ is right special in $w$.
 
 Suppose that $w$ has a right special factor $r$ longer than $h$. Since $w$ is trapezoidal, $w$ has at most one right special factor for each length (Proposition \ref{prop:trap}). Since the suffixes of a right special factor are right special factors, we have that $h$ must be a proper suffix of $r$. Since $r$ is right special in $w$, it has at least two occurrences in $w$ followed by different letters. This implies a non-prefix occurrence of $hx$ in $w$, against the definition of $h$.
 
$(2)\Rightarrow (3)$. Let $k$ be the longest repeated suffix of $w$. We first prove that $k$ is left special in $w$. Otherwise, $k$ appears in $w$ exactly twice, once as a prefix and once as a suffix of $w$. This implies that $k=h$, the longest repeated prefix of $w$ -- a contradiction, since by hypothesis $h$ is right special in $w$ and therefore it has at least a non-suffix occurrence. 

It remains to prove that $k$ is the longest left special factor of $w$. Since $w$ is trapezoidal, we have, by Proposition \ref{prop:trap}, $|w|=L_{w}+H_{w}=R_{w}+K_{w}$. Since by hypothesis the longest repeated prefix of $w$ is also the longest right special factor of $w$, we have $H_{w}=R_{w}$ and therefore $L_{w}=K_{w}$. Thus, the longest left special factor of $w$ has length equal to $L_{w}-1=K_{w}-1=|k|$.

$(3)\Rightarrow (1)$. Let $k$ be the longest repeated suffix of $w$ and let $ak$ and $bk$ be factors of $w$, for different letters $a$ and $b$.  Then we have that $\tilde{k}$ is the longest repeated prefix of $\tilde{w}$, and $\tilde{k}a$ and $\tilde{k}b$ are factors of $\tilde{w}$. This proves that the word $\tilde{w}$ is open. The claim then follows from Proposition \ref{prop:rev}.
\end{proof}

\begin{lemma}\label{cor:HRKL}
 Let $w$ be a trapezoidal word. If $w$ is open, then $H_{w}=R_{w}$ and $K_{w}=L_{w}$. If $w$ is closed, then $H_{w}=K_{w}$ and $L_{w}=R_{w}$.
\end{lemma}

\begin{proof}
The claim for open trapezoidal words follows directly from Proposition \ref{prop:openspe}. 
 
Suppose that $w$ is a closed trapezoidal word. Then $H_{w}=K_{w}$ (since $w$ is closed) and therefore $L_{w}=R_{w}$ (since $w$ is trapezoidal).
\end{proof}

Open trapezoidal words can be Sturmian (e.g.\ $w=aaabaa$) or not (e.g.\ $w=aaabab$). Closed trapezoidal words, instead, are always Sturmian, as shown in the following proposition.

\begin{proposition}\label{prop:cloStur}
Let $w$ be a trapezoidal word.  If $w$ is closed, then $w$ is Sturmian.
\end{proposition}

\begin{proof}
 Suppose that $w$ is not Sturmian. Then, by Theorem \ref{theor:dal}, $w=pq$, $p \in \Suff(\{\tilde{z}_{f}^{*}\})$, $q\in \Pref(\{z_{g}^{*}\})$, with $(f,g)$ being the pair of pathological factors of $w$ of minimal length, $K_{w}=|q|$ (and hence $R_{w}=|p|$) and the longest right special factor of $w$ is the prefix of $w$ of length $R_{w}-1$. 
 
Since $w$ is closed, we have, by Lemma \ref{cor:HRKL}, $H_{w}=K_{w}$. So the suffix $k$ of $q$ of length $K_{w}-1$ is also the longest repeated prefix of $w$ and appears in $w$ only as a prefix and as a suffix. 
  
If $R_{w}\ge K_{w}$, then $k$ is a prefix of the longest right special factor of $w$. This is a contradiction with the fact that $k$ appears in $w$ only as a prefix and as a suffix.  
  
If $R_{w}<K_{w}$, then $p$ is a prefix of $k$ and hence $p$ is a factor of $q$. This implies that $f$ is a factor of $q$, and therefore $q$ contains both $f$ and $g$ as factors. Hence $p$ would be non-Sturmian, a contradiction with Lemma \ref{lem:pk}.
\end{proof}

The result stated in Proposition \ref{prop:cloStur} can also be found in a paper of Bucci, de Luca and De Luca (\cite{BuDelDel09}, Proposition 3.6). 

As a corollary of Proposition \ref{prop:cloStur}, we have that every trapezoidal word is open or Sturmian. We therefore propose the following

\begin{problem}\label{prob:1}
 Give a characterization of open Sturmian words.
\end{problem}

Trapezoidal words, as well as Sturmian words, can be primitive (e.g.\ $w=aabaaa$) or not (e.g.\ $w=aabaab$). Open trapezoidal words (and in particular, then, non-Sturmian trapezoidal words) are always primitive.

\begin{lemma}\label{lem:oprim}
 Every open trapezoidal word is primitive. 
\end{lemma}

\begin{proof}
 Suppose that $w$ is not primitive. Let $w=u^{n}$, for a non-empty primitive word $u$ and an integer $n>1$. The longest repeated prefix of $w$ is therefore $u^{n-1}$, which is also a suffix of $w$. Moreover, by elementary combinatorics on words, $u^{n-1}$ cannot have internal occurrences in $w$. Hence $w$ is closed.
\end{proof}

The converse of Lemma \ref{lem:oprim} does not hold. Indeed, there exist trapezoidal words that are primitive but not open, e.g. $w=aabaa$.


We now focus on closed trapezoidal words and their special factors.

\begin{lemma}\label{lem:clo1}
Let $w$ be a closed trapezoidal word and let $u$ be  the longest left special factor of $w$. Then $u$ is also the longest right special factor of $w$ (and thus $u$ is a bispecial factor of $w$). Moreover, $u$ is a central word.  
\end{lemma}

\begin{proof}
Let $u$ be the longest left special factor of $w$. Hence there exist different letters $a,b\in \Sigma$ such that $au$ and $bu$ are factors of $w$. 

We claim that both $au$ and $bu$ occur in $w$ followed by some letter. Indeed, suppose the contrary. Then one of the words $au$ and $bu$, say $au$, appears in $w$ only as a suffix. Let $k$ be the longest repeated suffix of $w$. Since $u$ is a repeated suffix of $w$, we have $|k|\ge |u|$. If $|k|=|u|$, then $k=u$ and since $w$ is closed, $u$ appears in $w$ only as a prefix and as a suffix, against the hypothesis that $bu$ is a factor of $w$. So $|k|>|u|$ and therefore $au$ must be a suffix of $k$. This implies an internal occurrence of $au$ in $w$, a contradiction.

So, there exist letters $x,y$ such that $aux$ and $buy$ are factors of $w$. Now, we must have $x\neq y$, since otherwise $ux$ would be a left special factor of $w$ longer than $u$. Thus $u$ is right special in $w$. Since $w$ is closed, we have, by Lemma \ref{cor:HRKL}, $H_{w}=K_{w}$ and $L_{w}=R_{w}$, so $w$ cannot contain a right special factor longer than $u$. Thus $u$ is the longest right special factor of $w$.

By Proposition \ref{prop:cloStur}, $w$ is Sturmian. Since $u$ is a bispecial factor of $w$, and since a factor of a Sturmian word is a Sturmian word, in order to prove that $u$ is a central word it is sufficient to prove that $u$ is a palindrome. Suppose the contrary. So there exists a prefix $z$ of $u$ and a letter $a\in \Sigma$ such that $za$ is a prefix of $u$ and $b\tilde{z}$ is a suffix of $u$, for a letter $b$ different from $a$. Since $u$ is bispecial in $w$, this implies that $aza$ and $b\tilde{z}b$ are both factors of $w$. This implies that $w$ is not balanced, a contradiction with Proposition  \ref{prop:cloStur}.
\end{proof}

By Theorem \ref{theor:trappal}, trapezoidal palindromes coincide with Sturmian palindromes. Deep and interesting results on Sturmian palindromes can be found in \cite{DelDel05} and \cite{DelDel06a}. In particular, we want to recall the following

\begin{theorem}[\cite{DelDel06a}, Theorem 29]\label{theor:delucas}
 A palindrome $w\in \Sigma^{*}$ is Sturmian if and only if $\pi_{w}=R_{w}+1$.
\end{theorem}

The next theorem shows that Sturmian palindromes are all closed words. 

\begin{theorem}\label{theor:palclo}
 Let $w$ be a trapezoidal (Sturmian) palindrome. Then $w$ is closed.
\end{theorem}

\begin{proof}
By contradiction, suppose that $w$ is open. Then $h$, the longest repeated prefix of $w$,  is also the longest right special factor of $w$ (Proposition \ref{prop:openspe}). Since $w$ is a palindrome, we have that the longest repeated suffix of $w$ is $\tilde{h}$, the reversal of $h$. In particular, then, $K_{w}=H_{w}$.

By Lemma \ref{cor:HRKL}, $H_{w}=R_{w}$ and $K_{w}=L_{w}$. Thus we have $H_{w}=R_{w}=K_{w}=L_{w}=|w|/2$, since $w$ is trapezoidal (see Proposition \ref{prop:trap}). It follows that $w=hxx\tilde{h}$, for a letter $x\in \Sigma$. 

By Theorem \ref{theor:delucas}, the period of $w$ is $R_{w}+1=|hxx|$, so we have $\tilde{h}=h$.  Therefore, we have $w=hxxh$. 

Since $h$ is right special in $w$, there exists a letter $y\neq x$ such that $hy$ is a factor of $w$. Now, any occurrence of $hy$ in $w$ cannot be preceded by the letter $x$, since $h=\tilde{h}$ is the longest repeated suffix of $w$. Thus $w$ contains the factor $yhy$. 

Hence, $w$ contains both $hxx$ and $yhy$ as factors, and this contradicts the fact that $w$ is Sturmian.
\end{proof}

\begin{remark}
 The equivalence between trapezoidal palindromes and Sturmian palindromes (Theorem \ref{theor:trappal}) can also be derived as a consequence of Theorem \ref{theor:palclo}, Proposition \ref{prop:cloStur} and Proposition \ref{prop:sturmtrap}.
\end{remark}

From Theorem \ref{theor:palclo} and Lemma \ref{lem:clo1}, we derive the following characterization of the special factors of Sturmian palindromes.

\begin{corollary}\label{cor:sturpal}
 Let $w$ be a trapezoidal (Sturmian) palindrome. Then the longest left special factor of $w$ is also the longest right special factor of $w$ and it is a central word.
\end{corollary}

\begin{example}
 Let $w=aababaa$. The longest left special factor of $w$ is $aba$, which is also its longest right special factor and is a central word.
\end{example}

\section{Conclusions and Open Problems}

In this paper we distinguished trapezoidal words into two disjoint classes: open and closed. We derived some combinatorial and structural properties of these two classes of words.

Many further development directions can arise. For example, a challenging problem could be that of finding a characterization of open Sturmian words, that is, of Sturmian words for which the longest repeated prefix is also the longest right special factor (Problem \ref{prob:1}).

Another interesting problem concerns enumeration. Enumeration formulae for Sturmian words \cite{Mig91} and for primitive Sturmian words \cite{DelDel06} are known. To the best of our knowledge, an enumeration formula for trapezoidal words is not yet known. A possible direction for finding it could be enumerating open and closed trapezoidal words separately.

\section{Acknowledgments}

The author is grateful to anonymous referees for suggestions that greatly improved the presentation of the paper.

\bibliographystyle{eptcs}
\bibliography{trapezoidal}

\end{document}